\documentclass[conference,letterpaper]{IEEEtran}
\usepackage{balance}
\usepackage[utf8]{inputenc} 
\usepackage[T1]{fontenc}
\usepackage{url}
\usepackage{ifthen}
\usepackage{cite,enumitem}
\usepackage[cmex10]{amsmath} 
\usepackage{graphicx,mathtools}
\usepackage{amssymb,amsthm, verbatim}
\usepackage[colorlinks=true,citecolor=blue]{hyperref}
\usepackage{tcolorbox} 
\tcbuselibrary{listingsutf8} 

\definecolor{lightgray}{RGB}{224,224,224}

\newtheorem{theorem}{Theorem}

\newtheorem{corollary}{Corollary}

\newtheorem{remark}{Remark}
\newtheorem{proposition}{Proposition}
\newtheorem{lemma}{Lemma}

\begin{document}
\title{Maximal Achievable Service Rates of Codes and Connections to Combinatorial Designs}

\author{%
  \IEEEauthorblockN{Hoang Ly and Emina Soljanin}
  \IEEEauthorblockA{
  Rutgers University}
                    E-mail: \{\texttt{mh.ly;emina.soljanin}\}@rutgers.edu
}
\maketitle
\begin{abstract}
We investigate the service-rate region (SRR) of distributed storage systems that employ linear codes. We focus on systems where each server stores one code symbol, and a user recovers a data symbol by accessing any of its recovery groups, subject to per-server capacity limits. The SRR--the convex polytope of simultaneously achievable request rates--captures system throughput and scalability. We first derive upper and lower bounds on the maximum request rate of each data object. These bounds hold for all linear codes and depend only on the number of parity checks orthogonal to a particular set of codeword coordinates associated with that object, i.e., the equations used in majority-logic decoding, and on code parameters. 
We then check the bound saturation for 1) all non-systematic codes whose SRRs are already known and 2) systematic codes. For the former, we prove the bounds are tight. For systematic codes, we show that the upper bound is achieved whenever the supports of minimum-weight dual codewords form a 2-design. As an application, we determine the exact per-object demand limits for Hamming codes. Our framework offers a new perspective on addressing the SRR problem through the lens of combinatorial design theory. 
\end{abstract}
\section{Introduction}\label{intro}
Reliable and scalable data access is a key requirement in modern computing systems, particularly under dynamic and uneven demands. Traditional systems have relied mainly on basic replication strategies~\cite{SRR:journals/tit/AktasJKKS21}. The increasing disparity in data popularity has motivated hybrid redundancy schemes that combine erasure coding with replication. In this context, \emph{Service Rate Region} (SRR) has emerged as a foundational metric, quantifying the set of achievable request rates supported by a given redundancy configuration~\cite{SRR:journals/tit/AktasJKKS21,SRR:conf/isit/KazemiKS20}. The concept of SRR is closely related to load-balancing~\cite{LB:AktasFSW21}, and it generalizes the problem of batch codes~\cite{SRR:conf/isit/KazemiKSS20}. It is also connected to the problem of access balancing; see, e.g.,~\cite{MaxMinSum:DauM'18} and the follow-up work.

Previous work has characterized the SRRs of several prominent code families, including MDS, Simplex, first- and higher-order Reed–Muller codes, and locally recoverable codes; see~\cite{SRR:journals/tit/AktasJKKS21}, and more recently~\cite{SRR:lySV2025,SRR:preprint/arxiv/LySL25}. Surprisingly, however, the SRR of Hamming codes--among the earliest known linear block codes--has not yet been determined. In parallel, other studies have examined code parameters such as field size and block length needed to support a given SRR~\cite{CodeParameters:conf/isit/KilicRS24}. Despite these advances, several fundamental questions concerning the structure and limits of SRRs remain unresolved.

In distributed storage systems and SRR problems, each data object is retrieved through \emph{recovery sets} (or repair groups), which are collections of stored symbols that together enable the reconstruction of the desired object. This notion closely parallels that of \emph{majority-logic decoding} (MLD) in linear codes, where each message (or information) symbol is recovered using multiple (not necessarily disjoint) sets of codeword symbols. Each such set contributes a "vote" toward the actual value of the message symbol, and decoding succeeds if the majority of these votes remain uncorrupted in the presence of noise in the communication channel. Since errors in received codeword symbols affect all recovery sets that contain them, the effectiveness of MLD depends critically on the overlap structure among these sets. Codes that can be decoded using MLD are known as \emph{majority-logic decodable codes}. These codes are attractive in practice due to their inexpensive decoding circuitry and their potential to correct beyond their worst-case error bounds~\cite{Coding:books/PetersonW72}. 

The structure of the recovery sets, and thus the correction capability of MLD, is strongly tied to the combinatorics of the dual code, specifically to how codeword symbols appear in the supports of dual codewords. A fundamental result of Rudolph~\cite{MLD:Rudolph67} characterizes the error correction capability of MLD of codes whose dual codewords' supports exhibit certain combinatorial regularities. For some extensions of these results, see, e.g., \cite{MLD:Ng70,MLDComDesign:RahmanB75,LySoljanin2025OneStep}. This paper explicitly connects SRR analysis with MLD decoding and combinatorial design theory. We thus demonstrate how a linear code's combinatorial structure and error-correcting properties significantly impact its data access performance in storage systems. 

One of the central problems in SRR analysis is to design storage schemes that cover a set of target service-rate points using minimal resources~\cite{SRR:journals/tit/AktasJKKS21}. Among the most critical of these are the maximum individual service rates, the highest demand that can be served for each data object when the entire system is dedicated to that object, that is, when all other demands are set to zero. We refer to these points as the \emph{maximal achievable rates} for individual data objects. Characterizing them is particularly relevant in distributed storage systems, where skewed demand or object popularity is common~\cite{SRR:journals/tit/AktasJKKS21}.

This paper is organized as follows. Section II formally defines the SRR problem and the maximum achievable rate for each data object, and summarizes our main contributions. Section~III introduces majority-logic decoding and establishes bounds on the maximal achievable rate using orthogonal parity checks. In Section~IV, we specialize these results into systematic codes and identify the conditions under which the upper bound is achieved, drawing connections to combinatorial designs. Section~V concludes the paper. 
\section{Problem Statement}\label{sec:Problem_statement}
In this section, we first introduce the notation used throughout the paper, followed by basic definitions and relevant results from combinatorial design theory. The results we include are classical and well known, and can be found in standard references such as~\cite{CombiDesigns:books/daglib/Stinson04}. We then formally describe and formulate the SRR problem, and the maximal achievable rates of data objects, and conclude with a summary of the main results.
\subsection{Preliminaries and Notation}
The finite field over a prime or prime power \(q\) is denoted as \(\mathbb{F}_q\). A \(q\)-ary linear code \(\mathcal{C}\) with parameters \([n, k, d]_q\) is a \(k\)-dimensional subspace with minimum (Hamming) distance $d$ of the \(n\)-dimensional vector space \(\mathbb{F}_q^n\). Also, $d^{\perp}$ denotes the dual distance, i.e., minimum distance of the dual code $\mathcal{C}^{\perp}$. Hamming weight of a codeword \(\boldsymbol{x}\) is denoted as \(\textsf{w}(\boldsymbol{x})\). The symbols \(\boldsymbol{0}_k\) and \(\boldsymbol{1}_k\) denote the all-zero and all-one column vectors of length \(k\), respectively. Standard basis (column) vector with a one at position \(i\) and 0s elsewhere is represented by \(\boldsymbol{e}_i\), and its transpose by $\boldsymbol{e}_i^{\top}$. Moreover, $\boldsymbol{c}^j(\boldsymbol{G})$ and $\boldsymbol{h}^i(\boldsymbol{G})$ denote respectively the $j$-th column and the $i$-th row of the matrix $\boldsymbol{G}$ being considered, and $\boldsymbol{G}$ can be dropped when the context is clear. \(\mathrm{Supp}(\boldsymbol{x})\) denotes the support of a codeword \(\boldsymbol{x}\), and $x_j$ its $j$-th coordinate. The set of positive integers not exceeding \(i\) is denoted as \([i]\). Similarly, \([a, b]\) represents the set of integers between \(a\) and \(b\), where \(a, b \in \mathbb{N}\) and \(a < b\). The cardinality of a finite set $V$ is $|V|$. 

A pair \( (V, \mathcal{A}) \) is called a \( t\text{--}(n,k,\lambda) \) block design (or simply, a \( t \)-design) if it satisfies the following: \( V \) is a set of \( n \) elements (called \emph{points}), and \( \mathcal{A} \) is a collection of \( k \)-element subsets of \( V \) (called \emph{blocks}), such that every \( t \)-subset of \( V \) is contained in exactly \( \lambda \) blocks in \( \mathcal{A} \). A \emph{Steiner system} is a $t$-$(n,k,\lambda)$ block design with $t=2,\, \lambda=1$. A fundamental result involving designs is the "Lower-order balance", stated as follows:
\begin{lemma}[Lower-order balance; see, e.g.,{~\cite[Corollary 9.6]{CombiDesigns:books/daglib/Stinson04}}]
    Suppose that $(V, \mathcal{A})$ is a $t$-design, and $1 \le s \le t$. Then $(V, \mathcal{A})$ is a $s$-design.
    \label{thm:lower_balance}
\end{lemma}
That is, any $t$-design is automatically balanced at all lower levels; in particular, it is always a $(t-1)$-design.
\subsection{Service Rates of Codes}
Consider a storage system in which $k$ data objects $o_1, \hdots, o_k$ are stored on $n$ servers, labeled $1, \hdots, n$, using a linear $[n, k]_q$ code with generator matrix $\boldsymbol{G}_{k\times n} \in \mathbb{F}^{k\times n}_q$. Let $\boldsymbol{c}^j$ denote the $j$-th column of $\boldsymbol{G}$, for $j \in [n]$. A recovery set for the object $o_i$ is a set of stored symbols that can be used to recover $o_i$. For a storage system using $\boldsymbol{G}$, a set $R \subseteq [n]$ is a \textit{recovery set} for $o_i$ if $\boldsymbol{e}_i \in \text{span}(R) := \text{span}(\cup_{j \in R}\{\boldsymbol{c}^j\})$, i.e., the unit vector $\boldsymbol{e}_i$ can be recovered by a linear combination of the columns of $\boldsymbol{G}$ indexed by $R$. WLOG, we restrict our attention to \textit{minimal} recovery sets $R: \boldsymbol{e}_i \notin \text{span}(S),\, \forall\, S \subsetneq R$, which ensures that we never download more symbols than necessary to recover a data object.

Let $\mathcal{R}_i = \{R_{i, 1}, \hdots, R_{i, t_i}\}$ be the $t_i$ recovery sets for the object $o_i$. Denote by ${\lambda_i},\, \forall\, i \in [k]$ the arrival rate of requests (or demands) to download an object $o_i$, and as $\boldsymbol{\lambda} = (\lambda_1, \hdots, \lambda_k) \in \mathbb{R}_{\ge 0}^k$ the request rates for all objects. Let $\mu_{l} \in \mathbb{R}_{\ge 0}$ be the average rate at which the server $l\in [n]$ processes requests for data objects. We denote the serving rates (or serving capacities) of servers \( 1, \ldots, n \) by the vector \( \boldsymbol{\mu} = (\mu_1, \ldots, \mu_n) \), and assume that servers have \emph{uniform capacity}; that is, \( \mu_j = 1, \ \forall\, j \in [n] \), or equivalently, \( \boldsymbol{\mu} = (\boldsymbol{1}_n)^{\top} \).

A request for an object can be assigned to any of its possibly overlapping recovery sets. Let $\lambda_{i,j}$ be the portion of requests for object $o_i$ that are assigned to the recovery set $R_{i,j}, j \in [t_i]$. The service rate region (SRR) $\mathcal{S}(\boldsymbol{G}, \boldsymbol{\mu}) = \mathcal{S}(\boldsymbol{G}, \boldsymbol{1}) \subset \mathbb{R}_{\ge0}^k$ is defined as the set of all request vectors $\boldsymbol{\lambda}$ that can be served by a coded storage system with generator matrix $\boldsymbol{G}$ and serving rate vector $\boldsymbol{\mu}$. Alternatively, $\mathcal{S}(\boldsymbol{G}, \boldsymbol{\mu})$ can be defined as the set of all vectors $\boldsymbol{\lambda}$ for which there exist $\lambda_{i,j} \in \mathbb{R}_{\ge 0},\ i \in [k]$ and $j \in [t_i]$, satisfying the followings:
\begin{align}
        \sum_{j=1}^{t_i}\lambda_{i, j} & \;=\; \lambda_i, \quad \forall\, i \in [k],\\   
        \sum_{i=1}^{k}\sum_{\substack{j=1 \\ l\, \in\, R_{i, j}}}^{t_i}\lambda_{i, j} &\;\le\; \mu_{l}, \quad \forall\, l\, \in\, [n]\\
        \lambda_{i, j} &\;\in\; \mathbb{R}_{\ge 0}, \quad \forall\, i \in [k],\ j \in [t_i].
\end{align}
Constraints (1) guarantee that the demands for all objects are satisfied, and constraints (2) ensure that no server receives requests at a rate larger than its capacity. Such vectors $\boldsymbol{\lambda}$ form a \textit{polytope} in $\mathbb{R}_{\ge 0}^k$. Therefore, the SRR is also referred to as the service polytope. A key property of the service polytope is that it is convex, as illustrated in the following.
\begin{lemma}
[\hspace{-0.1mm}{\cite[Lem. 1]{SRR:conf/isit/KazemiKS20}}]\label{lem:convexity}
The service rate region $\mathcal{S}(\boldsymbol{G}, \boldsymbol{1}))$ is a non-empty, convex, closed, and bounded subset of $\mathbb{R}_{\ge 0}^k$.
\end{lemma}
An analysis of the geometric structure and properties of service rate regions, including their volume, can be found in~\cite{Service:journals/siaga/AlfaranoKRS24}.
\subsection{Maximal Achievable Individual Service Rates}
For each $\ell\in[k]$ define the \emph{coordinate-wise maximum}
\[
   \lambda_{\ell}^{\max}\; :=\;
   \max\bigl\{\lambda_\ell \mid \boldsymbol{\lambda}\in\mathcal S(\boldsymbol G,\boldsymbol 1)\bigr\},
\]
i.e., the largest request rate for object $o_\ell$ achievable while other objects may also receive traffic. Next, we define the \emph{axis intercept}
\[
   \lambda_{\ell}^{\mathrm{int}}\;:=\;
   \max\bigl\{\gamma\in\mathbb R_{\ge0} \mid \gamma\,\boldsymbol e_\ell\in\mathcal S(\boldsymbol G,\boldsymbol 1)\bigr\},
\]
i.e., the maximum rate for $o_\ell$ achievable when all other demands are zero. It was proved in~\cite{SRR:lySV2025} that
\(\lambda_{\ell}^{\mathrm{int}}=\lambda_{\ell}^{\max}\). Consequently, dedicating the entire system to object~$o_\ell$ permits serving at most \(\lambda_{\ell}^{\max}\) requests. We refer to this quantity as the \emph{maximal achievable rate} (or maximal achievable rate) for object \( \ell \). 

We form the simplex
\[
   \Delta_{\max}\;:=\;
   \operatorname{conv}\bigl\{
      \boldsymbol{0}_k,\;
      \lambda_{1}^{\max}\boldsymbol e_1,\dots,
      \lambda_{k}^{\max}\boldsymbol e_k
   \bigr\},
\]
which is the convex hull of the origin and the $k$ axis-intercept vertices. Because $\mathcal S(\boldsymbol G,\boldsymbol 1)$ is convex (Lemma~\ref{lem:convexity}), \(\Delta_{\max}\subseteq\mathcal S(\boldsymbol G,\boldsymbol 1)\). No larger axis-aligned simplex lies inside the SRR, so we call $\Delta_{\max}$ the \emph{maximal achievable simplex}, in which every point is achievable.

\subsection{Summary of Contributions}
This paper focuses on the maximal achievable rate for each individual object, denoted \( \lambda_j^{\max} \), and establishes both bounds and achievability conditions for these rates. 
We also reveal connections between the SRR problem, MLD, and combinatorial design theory. Briefly, the contributions of this work are the following:
\begin{enumerate}
    \item We establish upper and lower bounds on the maximal achievable rate for each data object in a system encoded by a generator matrix \( \boldsymbol{G} \), expressed in terms of the number of parity checks orthogonal to its minimal recovery set and the parameters of the underlying code. Unlike previously established bounds that apply only to specific codes (e.g., Reed--Muller codes) or code families (e.g., systematic codes), our result holds for all binary linear codes (Theorem~\ref{thm:general_bound_demand}).

    \item We specialize these bounds to the case where \( \boldsymbol{G} \) is systematic and show that the upper bound is achieved when the minimum-weight dual codewords satisfy a specific structural condition. This result reveals a direct connection between the maximal achievable rate in the SRR problem and the error-correction capability of Majority-logic decoding applied to the underlying code (Theorem~\ref{thm:systematic} and Corollary~\ref{coro:2desgin}).
    \item We verify that the bounds are tight for all codes whose SRR has previously been characterized in the literature (Remark~\ref{remark:verify} and Section~\ref{sec:Systematic}.B).
    \item We extend the analysis to binary Hamming codes, providing the first analysis of their service rate region (Theorem~\ref{thm:Hamming_demand}). 
\end{enumerate}

\section{Majority-logic Decoding}
In this section, we provide essential background on MLD, which will be critical for and repeatedly used in the sections that follow. Readers familiar with MLD may skip this section.

We adopt standard notations used in~\cite{Coding:books/PetersonW72,Coding:books/MacWilliamsS77}. MLD is simple to implement, and also effective for several classes of codes, including Reed--Muller codes~\cite{Coding:books/PetersonW72}, and some finite geometry codes~\cite{MLDecoding:journals/tit/CruzW21}. We illustrate its fundamental principle using the \([15, 4, 8]\) Simplex code $\mathcal{C}$, whose generator matrix is
\[
\boldsymbol{G}_S =
\scalebox{0.93}{$
\left[\begin{array}{@{\hskip 0.2pt}*{15}{c}@{\hskip 0.4pt}}
0 & 0 & 0 & 0 & 0 & 0 & 0 & 1 & 1 & 1 & 1 & 1 & 1 & 1 & 1\\
0 & 0 & 0 & 1 & 1 & 1 & 1 & 0 & 0 & 0 & 0 & 1 & 1 & 1 & 1\\
0 & 1 & 1 & 0 & 0 & 1 & 1 & 0 & 0 & 1 & 1 & 0 & 0 & 1 & 1\\
1 & 0 & 1 & 0 & 1 & 0 & 1 & 0 & 1 & 0 & 1 & 0 & 1 & 0 & 1
\end{array}\right],
$}
\]
and a corresponding parity-check matrix is
\[
\boldsymbol{H}_S =
\scalebox{0.93}{$
\left[\begin{array}{@{\hskip 0.3pt}*{15}{c}@{\hskip 0.2pt}}
1 & 0 & 0 & 0 & 0 & 0 & 0 & 0 & 0 & 0 & 0 & 1 & 1 & 0 & 0\\
1 & 1 & 0 & 0 & 0 & 0 & 0 & 0 & 0 & 0 & 0 & 0 & 1 & 1 & 0\\
1 & 1 & 1 & 0 & 0 & 0 & 0 & 0 & 0 & 0 & 0 & 0 & 0 & 0 & 0\\
1 & 0 & 0 & 1 & 0 & 0 & 0 & 0 & 0 & 0 & 0 & 0 & 0 & 1 & 0\\
1 & 0 & 0 & 0 & 1 & 0 & 0 & 0 & 0 & 0 & 0 & 0 & 1 & 0 & 1\\
1 & 0 & 0 & 0 & 0 & 1 & 0 & 0 & 0 & 0 & 0 & 1 & 0 & 0 & 1\\
1 & 0 & 0 & 0 & 0 & 0 & 1 & 0 & 0 & 0 & 0 & 0 & 0 & 0 & 1\\
1 & 0 & 0 & 0 & 0 & 0 & 0 & 1 & 0 & 0 & 0 & 1 & 1 & 1 & 1\\
1 & 0 & 0 & 1 & 0 & 0 & 1 & 0 & 1 & 0 & 0 & 0 & 1 & 0 & 0\\
1 & 0 & 0 & 0 & 0 & 0 & 0 & 0 & 0 & 1 & 0 & 1 & 0 & 1 & 1\\
1 & 0 & 0 & 0 & 0 & 0 & 0 & 0 & 0 & 0 & 1 & 0 & 0 & 1 & 1
\end{array}\right].
$}
\]
The code maps information vector $\boldsymbol{a} = [a_1, a_2, a_3, a_4]$ into codeword $\boldsymbol{x} = [x_1, x_2, \dots , x_{13}, x_{14}, x_{15}]$ via $\boldsymbol{x} = \boldsymbol{a}\cdot \boldsymbol{G}_S$. Note that for this specific generator matrix $\boldsymbol{G}_S$, the first column is $\boldsymbol{c}^1 = [0, 0, 0, 1]^\top = \boldsymbol{e}_4$. Therefore, the fourth message symbol $a_4$ is directly mapped to the first codeword symbol, i.e., 
\begin{align}\label{eq:MessageSymbol}
a_4 = \boldsymbol{a}\cdot \boldsymbol{e}_4 = \boldsymbol{a}\cdot \boldsymbol{c}^1 = x_1.
\end{align}
The parity-check matrix \( \boldsymbol{H}_S \) defines constraints on any codeword \( \boldsymbol{x} \in \mathcal{C} \) via the relation \( \boldsymbol{H}_S \cdot \boldsymbol{G}_S^\top = \boldsymbol{0}_{11 \times 4} \). Since \( \boldsymbol{H}_S \) serves as a generator matrix for the dual code \( \mathcal{C}^\perp \), its row space contains all dual codewords of \( \mathcal{C} \). The dual code has dimension \( 15 - 4 = 11 \), so there are \( 2^{11} \) such codewords, each corresponding to a parity-check equation on codewords of \( \mathcal{C} \). Consider the four rows 1, 3, 4, 7 of $\boldsymbol{H}_S$ that define the following parity checks (over the binary field $\mathbb{F}_2$)
\begin{align}
\boldsymbol{G}_S\cdot(\boldsymbol{h}^1)^{\top} = \boldsymbol{0}_4, \, \text{or equally, }\, \boldsymbol{c}^{1}+\boldsymbol{c}^{12}+\boldsymbol{c}^{13} \;& =\; \boldsymbol{0}_4\notag\\ \text{or, }\, x_1 + x_{12} + x_{13} &= 0 \label{eq:pc1}\\[2pt]
\boldsymbol{G}_S\cdot(\boldsymbol{h}^3)^{\top} = \boldsymbol{0}_4, \, \text{or equally, }\, \boldsymbol{c}^{1}+\boldsymbol{c}^{2}+\boldsymbol{c}^{3} \;&=\; \boldsymbol{0}_4\notag\\ \text{or, }\, x_1 + x_{2} + x_{3} &= 0 \label{eq:pc3}\\[2pt]
\boldsymbol{G}_S\cdot(\boldsymbol{h}^4)^{\top} = \boldsymbol{0}_4, \, \text{or equally, }\, \boldsymbol{c}^{1}+\boldsymbol{c}^{4}+\boldsymbol{c}^{14} &= \;\boldsymbol{0}_4 \notag\\
\text{or, }\, x_1 + x_{4} + x_{14} &= 0 \label{eq:pc4}\\[2pt]
\boldsymbol{G}_S\cdot(\boldsymbol{h}^7)^{\top} = \boldsymbol{0}_4, \, \text{or equally, }\, \boldsymbol{c}^{1}+\boldsymbol{c}^{7}+\boldsymbol{c}^{15} &= \;\boldsymbol{0}_4 \notag\\
\text{or, }\, x_1 + x_{7} + x_{15} &= 0 \label{eq:pc7}.
\end{align}
From the four equations~\eqref{eq:pc1}--\eqref{eq:pc7} we can derive five independent estimates, or "votes", for the symbol $a_4$:
\begin{align}
a_4 &\;=\; x_1 \quad &(\text{from } \eqref{eq:MessageSymbol}) \label{eq:vote0} \\
&\;=\; x_{12} + x_{13} \quad &(\text{from } \eqref{eq:pc1}) \label{eq:vote1} \\
&\;=\; x_{2} + x_{3} \quad &(\text{from } \eqref{eq:pc3}) \label{eq:vote2} \\
&\;=\; x_4 + x_{14} \quad &(\text{from } \eqref{eq:pc4}) \label{eq:vote3} \\
&\;=\; x_7 + x_{15} \quad &(\text{from } \eqref{eq:pc7}) \label{eq:vote4}
\end{align}
Observe equations \eqref{eq:pc1},~\eqref{eq:pc3}, \eqref{eq:pc4}, and \eqref{eq:pc7}. The symbol $x_1$ participates in all four checks. Crucially, among the other symbols involved ($\{x_{12}, x_{13}\}, \{x_{2}, x_{3}\}$, $\{x_4, x_{14}\}$, $\{x_7, x_{15}\}$), each symbol appears in exactly one of these four checks. This property is called orthogonality (defined formally below).

If the received vector $\boldsymbol{y}$ differs from the transmitted codeword $\boldsymbol{x}$ by at most two errors, the orthogonality of the parity checks ensures that each error affects at most one of the check sums on the right-hand side of \eqref{eq:vote0}–\eqref{eq:vote4}.

\medskip
\noindent
\textit{Example.} Suppose that $y_4$ and $y_2$ are flipped, and all other symbols are received correctly. Then
\begin{itemize}
  \item $y_4 + y_{14}$ and $y_2 + y_{3}$ will disagree with $y_1$;
  \item $y_{12} + y_{13}$ and $y_7 + y_{15}$ will still match $y_1$.
\end{itemize}

A majority vote among the five estimates 
\[
y_1,\quad y_{12} + y_{13},\quad y_2 + y_{3},\quad y_4 + y_{14},\quad y_7 + y_{15}
\]
can correctly recover the transmitted symbol $a_4$. In the event of a tie, the decoder may favor $0$.

The same idea applies to other message symbols, each using an appropriate subset of parity checks drawn from the full set of $2^{11}$ available.

Since all these parity checks and majority votes can be computed in parallel over the received codeword, the procedure is referred to as \textbf{one-step majority-logic decoding}--in contrast to the sequential, multi-step Reed decoding algorithm~\cite{Coding:books/PetersonW72}.

\medskip
\noindent
Formally, we define orthogonality in terms of the index sets involved in the parity checks. For any index set \( \mathcal{I} \subseteq [n] \), its \emph{incidence vector} \( \boldsymbol{\chi}(\mathcal{I}) \in \{0,1\}^n \) is defined as
\[
\boldsymbol{\chi}(\mathcal{I})_j \;=\;
\begin{cases}
1 & \text{if } \boldsymbol{c}^j \in \mathcal{I}, \\
0 & \text{otherwise},
\end{cases}
\]
where \( \boldsymbol{c}^j \) denotes the \( j \)-th column of the generator matrix~$\boldsymbol{G}$.

A set of indices \( \mathcal{I} \subseteq [n] \) defines a parity-check equation if the sum of the corresponding columns is zero, i.e.,
\[
\sum_{j \in \mathcal{I}} \boldsymbol{c}^j \;=\; \boldsymbol{G} \cdot \boldsymbol{\chi}(\mathcal{I})^\top \;=\; \boldsymbol{0}_k,
\]
which generalizes equations~\eqref{eq:pc1}–\eqref{eq:pc7}.

Consider $J$ parity-check equations defined by index sets $\mathcal{I}_1, \mathcal{I}_2, \dots, \mathcal{I}_J$. These checks are said to be \textit{orthogonal} on a set of indices $\mathcal{O}$ if every index in $\mathcal{O}$ is involved in every check $\mathcal{I}_i$, and no index \textit{outside} $\mathcal{O}$ is involved in more than one check $\mathcal{I}_i$. Mathematically,
\[
\mathcal{I}_i \;\cap\; \mathcal{I}_j \;=\; \mathcal{O}, \quad \text{for all distinct } i, j \in [J].
\]
In our example, the four chosen parity checks correspond to index sets $\mathcal{I}_1 = \{1, 12, 13\}$, $\mathcal{I}_2 = \{1, 2, 3\}$, $\mathcal{I}_3 = \{1, 4, 14\}$, and $\mathcal{I}_4 = \{1, 7, 15\}$. We have $J=4$, and these sets are orthogonal on $\mathcal{O} = \{1\}$ since $\mathcal{I}_i \cap \mathcal{I}_j = \{1\}$ for $i \neq j$.

Orthogonal check sums are central to MLD, as captured by the following result.
\begin{proposition}[see e.g., {\cite[Ch.~13]{Coding:books/MacWilliamsS77}}]
If there are $J$ parity checks orthogonal on every coordinate, the code can correct $\lfloor J/2 \rfloor$ or fewer errors using one-step MLD.
\label{prop:MJDecoding}
\end{proposition}
\textit{Beyond the guaranteed bound.} MLD can often succeeds even when more than $\lfloor J/2 \rfloor$ errors occur~\cite{Coding:books/MacWilliamsS77}. For instance, if exactly four bits $y_2$, $y_3$, $y_4$, $y_{14}$ are flipped, all the parity-check votes for $a_4$ remain correct.

\begin{remark}
Note that if $J$ parity checks are \emph{orthogonal} on a set $\mathcal{O}$, then any subset of size $s \le J$ is also orthogonal on $\mathcal{O}$. However, Proposition~\ref{prop:MJDecoding} shows that error-correction capability improves with a larger number of orthogonal checks. Thus, for a given coordinate, the key quantity governing MLD is the maximum number~$J$ of parity checks orthogonal on that coordinate. Throughout this paper, any reference to parity checks orthogonal to a coordinate implicitly refers to a \emph{maximal} such set--i.e., one of the largest possible size.
\end{remark}
\section{MLD and SRR's Axis Intercept Points}\label{sec:SRR_MLD}
We show how the MLD orthogonal checks correspond to recovery sets for data objects in the SRR problem, enabling tight bounds on the maximum achievable rate for each object. We state the concept for binary codes; generalization to linear codes over any finite field is straightforward.

Now consider $\mathcal{O}$, a smallest recovery set for object ${\ell} \in [k]$. In other words, $\mathcal{O}$ is a set of minimal cardinality such that
\[
\sum\limits_{j \in \mathcal{O}}\,\boldsymbol{c}^j \,=\, \boldsymbol{e}_{\ell}.
\]
Let $a = |\mathcal{O}|$ denote the size of this smallest recovery set. Let $\mathcal{I}_1, \dots, \mathcal{I}_{J_{\mathcal{O}}}$ be the index sets of $J_{\mathcal{O}}$ parity-check sums such that these checks are orthogonal on the set $\mathcal{O}$. (Here, $J_{\mathcal{O}}$ denotes the maximum number of such checks orthogonal on this specific set $\mathcal{O}$.) Recall that for any $\mathcal{I}_{i}$, the definition of the parity check implies 
\[
\sum_{j \in \mathcal{I}_i}\, \boldsymbol{c}^j \,=\, \boldsymbol{G}\cdot\boldsymbol{\chi}(\mathcal{I}_i)^{\top} \,=\, \boldsymbol{0}_k.
\]

Then, we have for each $i \in [J_{\mathcal{O}}]$
\[
\sum\limits_{j \in \mathcal{I}_i\setminus\mathcal{O}}\boldsymbol{c}^j \,=\, \sum\limits_{j \in \mathcal{I}_i}\boldsymbol{c}^j - \sum\limits_{j \in \mathcal{O}}\boldsymbol{c}^j \,=\, \boldsymbol{0}_k - \boldsymbol{e}_{\ell} \,=\, \boldsymbol{e}_{\ell} \ \ \ (\text{over } \mathbb{F}_2).
\]

In other words, the sets \( \mathcal{R}_i = \mathcal{I}_i \setminus \mathcal{O} \) for \( i = 1, \dots, J_{\mathcal{O}} \) form recovery sets for object \( o_\ell \). Including \( \mathcal{O} \) itself—which is also a recovery set for \( o_\ell \)—we obtain \( J_{\mathcal{O}} + 1 \) recovery sets in total. Moreover, by the orthogonality condition \( \mathcal{I}_i \cap \mathcal{I}_k = \mathcal{O} \) for \( i \ne k \), the sets \( \mathcal{R}_i \) are pairwise disjoint by construction.

Thus, we have constructed \( J_{\mathcal{O}} + 1 \) pairwise disjoint recovery sets for \( o_\ell \): the original set \( \mathcal{O} \) and the \( J_{\mathcal{O}} \) sets \( \mathcal{R}_1, \dots, \mathcal{R}_{J_{\mathcal{O}}} \). This construction leads to the following theorem.

\begin{theorem}\label{thm:general_bound_demand}
Consider a storage system encoded by a generator matrix \( \boldsymbol{G} \). Let \( \mathcal{O} \) be a smallest recovery set for object \( o_\ell \), with \( \ell \in [k] \), and let \( a = |\mathcal{O}| \) denote its size. Let \( J_{\mathcal{O}} \) be the maximum number of parity-check sums orthogonal on \( \mathcal{O} \). Then the maximal achevable rate \( \lambda_\ell^{\max} \) for object \( o_\ell \) satisfies
\[
1 + J_{\mathcal{O}} \,\le\, \lambda_{\ell}^{\max} \,\le\, 1 + \frac{n - a}{\max\{d^\perp - a,\ a\}} \,\le\, 1 + \frac{n - a}{d^\perp - a},
\]
where \( d^{\perp} \) is the minimum distance of the dual code \( \mathcal{C}^{\perp} \).
\end{theorem}

\begin{proof}
The lower follows from the above argument. We now prove the upper bound.
\\[2ex]
\textbf{Step 1: Every recovery set other than \( \mathcal{O} \) has size at least \( d^\perp - a \).}

Let \( \mathcal{O} \) be a minimal recovery set for \( o_\ell \), with size \( a = |\mathcal{O}| \), and let \( \mathcal{O}' \ne \mathcal{O} \) be any other recovery set for \( o_\ell \). Since both \( \mathcal{O} \) and \( \mathcal{O}' \) recover \( o_{\ell} \), we have
\[
\sum_{j \in \mathcal{O}} \boldsymbol{c}^j \,=\, \boldsymbol{e}_{\ell} \,=\, \sum_{j \in \mathcal{O}'} \boldsymbol{c}^j.
\]
Adding these equations over \( \mathbb{F}_2 \) yields
\[
\sum_{j \in \mathcal{O}} \boldsymbol{c}^j + \sum_{j \in \mathcal{O}'} \boldsymbol{c}^j \,=\, \boldsymbol{0}_k.
\]
Let \( \mathcal{O}_1 := (\mathcal{O} \cup \mathcal{O}') \setminus (\mathcal{O} \cap \mathcal{O}') \) be the symmetric difference of \( \mathcal{O} \) and \( \mathcal{O}' \). Then
\[
\sum_{j \in \mathcal{O}_1} \boldsymbol{c}^j \,=\, \boldsymbol{0}_k,
\]
which implies that the incidence vector \( \boldsymbol{\chi}(\mathcal{O}_1) \in \mathbb{F}_2^n \) satisfies
\[
\boldsymbol{G} \cdot \boldsymbol{\chi}(\mathcal{O}_1)^\top \,=\, \boldsymbol{0}_k,
\]
so \( \boldsymbol{\chi}(\mathcal{O}_1) \in \mathcal{C}^\perp \). Since \( \mathcal{O} \ne \mathcal{O}' \), the set \( \mathcal{O}_1 \) is nonempty, and thus \( \boldsymbol{\chi}(\mathcal{O}_1) \) is a nonzero codeword in \( \mathcal{C}^\perp \) with weight at least \( d^\perp \). One has
\[
|\mathcal{O}_1| \,=\, \textsf{w}(\boldsymbol{\chi}(\mathcal{O}_1)) \,\ge\, d^\perp.
\]
By the identity
\[
|\mathcal{O}_1| \,=\, |\mathcal{O}| + |\mathcal{O}'| - 2|\mathcal{O} \cap \mathcal{O}'|,
\]
we obtain
\[
|\mathcal{O}'| \,\ge\, |\mathcal{O}_1| - |\mathcal{O}| \,\ge\, d^\perp - a.
\]
Therefore, every recovery set distinct from \( \mathcal{O} \) must have size at least \( d^\perp - a \). Since \( \mathcal{O} \) has size \( a \) and is by definition the smallest recovery set, it follows that the size of any other recovery set is at least \( \max\{d^\perp - a,\, a\} \).
\\[2ex]
\textbf{Step 2: Upper bound via capacity constraint.}
To support a demand rate \( \lambda_{\ell} = \lambda_{\ell}^{\max} \), the system must allocate server capacity across recovery sets for object \( o_{\ell} \). As shown in Step 1, the object can be recovered from
\begin{itemize}
    \item one recovery set \( \mathcal{O} \) of size \( a \), and
    \item other recovery sets, each of size at least \( \max\{d^{\perp} - a,\, a\} \ge a \).
\end{itemize}

To minimize total server usage, we consider the following assignment: allocate one unit of request to the smallest recovery set \( \mathcal{O} \), which consumes \( a \) servers; assign the remaining \( \lambda_{\ell} - 1 \) units of request to other recovery sets, each of size at least \( \max\{d^\perp - a,\, a\} \). Therefore, the total number of servers involved is at least
\[
1\cdot a + (\lambda_{\ell} - 1)\cdot\max\{d^\perp - a,\ a\}.
\]

Since this request is achievable, the total required server capacity must not exceed the total available capacity \( n \). Thus,
\[
a + (\lambda_{\ell} - 1)\cdot\max\{d^\perp - a,\ a\} \,\le\, n.
\]

Solving for \( \lambda_{\ell} \) yields
\[
\lambda_{\ell} \,\le\, 1 + \frac{n - a}{\max\{d^\perp - a,\, a\}} \,\le\, 1 + \frac{n - a}{d^\perp - a}.
\]

Since this holds for all achievable \( \lambda_{\ell} \), it applies in particular to \( \lambda_{\ell}^{\max} \). Together with the lower bound \( \lambda_{\ell}^{\max} \ge 1 + J_{\mathcal{O}} \), this completes the proof.
\end{proof}

\begin{remark}
Theorem~\ref{thm:general_bound_demand} immediately yields an upper bound on \( J_{\mathcal{O}} \), i.e., the number of parity-check sums orthogonal on a given set of codeword coordinates \( \mathcal{O} \) of size \( a \)
\[
   J_{\mathcal{O}} \;\le\; \frac{n - a}{\max\{a, \, d^{\perp} - a\}}.
\]
\end{remark}

\begin{remark}\label{remark:verify}
We show that the upper bound in Theorem~\ref{thm:systematic} is achieved for nonsystematic codes whose SRR is known up to now, which are Reed-Muller codes and nonsystematic MDS codes.
\begin{itemize}
    \item \textbf{Reed--Muller codes.} Consider the $\mathrm{RM}(r, m)$ code of length $n = 2^m$ and order $r$. Its dual is the $\mathrm{RM}(m - r - 1, m)$ code, which has minimum distance $d^\perp = 2^{r+1}$~\cite{Coding:books/MacWilliamsS77}. Substituting into the bound, we obtain
    \[
    \lambda_{\ell}^{\max} \,\le\, 1 + \frac{2^m - a}{2^{r+1} - a},
    \]
    where $a$ is the size of the smallest recovery set for object $o_{\ell}$. It was shown in~\cite{SRR:preprint/arxiv/LySL25} that this bound is achieved with equality, establishing the tightness of the bound for Reed--Muller codes.
    \item \textbf{Non-systematic MDS codes.} Consider a non-systematic \( [n,\, k,\, n - k + 1] \) MDS code ($n\ge k+1$), whose dual is an \( [n,\, n - k,\, k + 1] \) MDS code. In such codes, each recovery set has size \( k \). Applying the general bound yields
\[
\lambda_j^{\max} \,\le\, 1 + \frac{n - k}{k + 1 - k} \,=\, 1 + (n - k), \quad \forall\, j \in [k],
\]
which significantly overestimates the true value.

A tighter bound is obtained by refining the denominator to account for the actual recovery set size
\[
\lambda_j^{\max} \,\le\, 1 + \frac{n - k}{\max\{k + 1 - k,\, k\}} \,=\, 1 + \frac{n - k}{k} \,=\, \frac{n}{k}.
\]
This bound is known to be tight and achievable for non-systematic MDS codes, as established in~\cite{SRR:lySV2025, SRR:journals/tit/AktasJKKS21}.

We also observe that our result offers an alternative estimate, where the maximal achievable requests for individual objects are bounded in terms of the dual distance \( d^{\perp} \) rather than the minimum distance \( d_{\min} \) of the original code. In contrast, the bound from~\cite[Corollary 2]{SRR:conf/isit/KazemiKS20} states that for any linear code,
\[
\left\lceil \min_j \{ \lambda_j^{\max} \} \right\rceil \,\le\, d_{\min}.
\]
For instance, in the case of \( [n,\, k,\, n - k + 1] \) MDS codes, this yields the bound \( \lambda_j^{\max} \le n - k + 1 \), which is strictly weaker than the bound \( \lambda_j^{\max} \le n/k \) whenever \( n \ge k \).
\end{itemize}
\end{remark}

\section{Application to Systematic codes and Connection to Combinatorial Designs}\label{sec:Systematic}
In this section, we apply the results of the previous section to systems where the generator matrix \( \boldsymbol{G} \) is systematic, so that each object has a recovery set of size~1. We then propose a sufficient condition under which the upper bound on the maximal achievable rate is attained, and demonstrate that this condition holds for several classes of systematic codes, recovering known results and establishing new ones.

\subsection{Bounds for Maximal achievable rates in Systematic Codes}

We apply the results from Section~\ref{sec:SRR_MLD} to cases when $\boldsymbol{G}$ is systematic, which is of the form
\begin{align}\label{eq:systematic_generator}
    \boldsymbol{G}_{k\times n} \,=\, \left[\boldsymbol{I}_{k} \ |\ \boldsymbol{P}_{k\times (n-k)}\right],
\end{align}
where the first $k$ columns form the identity matrix $\boldsymbol{I}_k$, and the remaining $n-k$ columns provide redundancy. In this setting, each column \( \boldsymbol{c}^{\ell} = \boldsymbol{e}_{\ell} \) for \( \ell \in [k] \), so the smallest recovery set for object \( o_{\ell} \) is simply \( \mathcal{O} = \{\ell\} \), consisting of the systematic column \( \boldsymbol{c}^{\ell} \) itself, and has size 1. In this case, \( J_{\mathcal{O}} \) denotes the maximum number of parity-check sums that are orthogonal on coordinate \( \ell \)--that is, each parity check includes \( x_\ell \) and intersects the others only at \( x_\ell \). In this case, we will call this quantity instead by \( J_{\ell} := J_{\mathcal{O}} \). We assume that the dual distance satisfies \( d^{\perp} \ge 2 \), so that \( \max\{d^{\perp} - a,\ a\} = \max\{d^{\perp} - 1,\ 1\} = d^{\perp} - 1 \).

Consider the systematic \( [15, 11, 3] \) Hamming code \( \mathcal{C} \), generated by a systematic matrix \( \boldsymbol{G}_H \) of the form~\eqref{eq:systematic_generator}, and a storage system encoded using \( \boldsymbol{G}_H \). The dual code \( \mathcal{C}^\perp \) is the \( [15, 4, 8] \) Simplex code, generated by \( \boldsymbol{G}_S \) as introduced in the previous section. In this setting, the upper bound from Theorem~\ref{thm:general_bound_demand} becomes

\[
\lambda_1^{\max} \,\leq\, 1 + \frac{15 - 1}{8 - 1} = 1 + \frac{14}{7} = 3.
\]

We now show that this bound is indeed achievable. Consider the set of all minimum-weight dual codewords \( \boldsymbol{h} \in \mathcal{C}^\perp \) (in this case, those of weight 8) with \( h_1 = 1 \). Their supports are listed below:
\begin{align}
\mathrm{Supp}(\boldsymbol{h}^1) = \{1, 3, 5, 7, 9, 11, 13, 15\}\notag\\
\mathrm{Supp}(\boldsymbol{h}^2) = \{1, 2, 5, 6, 9, 10, 13, 14\}\notag\\
\mathrm{Supp}(\boldsymbol{h}^3) = \{1, 3, 4, 6, 9, 11, 12, 14\}\notag\\
\mathrm{Supp}(\boldsymbol{h}^4) = \{1, 2, 4, 7, 9, 10, 12, 15\}\notag\\
\mathrm{Supp}(\boldsymbol{h}^5) = \{1, 3, 5, 7, 8, 10, 12, 14\}\notag\\
\mathrm{Supp}(\boldsymbol{h}^6) = \{1, 2, 5, 6, 8, 11, 12, 15\}\notag\\
\mathrm{Supp}(\boldsymbol{h}^7) = \{1, 3, 4, 6, 8, 10, 13, 15\}\notag\\
\mathrm{Supp}(\boldsymbol{h}^8) = \{1, 2, 4, 7, 8, 11, 13, 14\}\notag
\end{align}
Each $\boldsymbol{h}^i$ satisfies $\boldsymbol{G}_H \cdot (\boldsymbol{h}^i)^\top = \boldsymbol{0}_{11}$, and since $\boldsymbol{G}_H$ is systematic, the first column of $\boldsymbol{G}_H$ is $\boldsymbol{c}^1(\boldsymbol{G}_H) = \boldsymbol{e}_1$, we obtain
\begin{align}\label{eq:sumColumns}
\sum_{j\, \in\, \mathrm{Supp}(\boldsymbol{h}^i) \setminus \{1\}} &\boldsymbol{c}^j(\boldsymbol{G}_H)\nonumber \\ & = \sum_{j\, \in\, \mathrm{Supp}(\boldsymbol{h}^i)} \boldsymbol{c}^j(\boldsymbol{G}_H) - \boldsymbol{c}^1(\boldsymbol{G}_H)\notag\\
\,& =\, \boldsymbol{0} - \boldsymbol{e}_1 \,=\, \boldsymbol{e}_1, \quad \forall\, i = 1, \dots, 8,
\end{align}
which implies that each set $\mathrm{Supp}(\boldsymbol{h}^i) \setminus \{1\}$ is a recovery set for object $o_1$. Notably, each coordinate $j \in \{2, 3, \dots, 15\}$ appears in precisely 4 of these eight sets, forming a $1$-$(14, 7, 4)$ design.

Therefore, we can assign demand for object $o_1$ to its recovery sets as follows:
\[
\scalebox{0.95}{$
\begin{cases}
    \text{Assign } \lambda_1 = 1 \text{ to the singleton set } \boldsymbol{c}^1 = \boldsymbol{e}_1, \\
    \text{Assign } \lambda_{-1} = \frac{1}{4} \text{ to each of 8 sets } \mathrm{Supp}(\boldsymbol{h}^i) \setminus \{1\},\ i \in [8].
\end{cases}
$}
\]
This assignment ensures that the total request served by each server is exactly $\frac{1}{4}\cdot 4 = 1$, and the total demand for $o_1$ is
\[
1 + \frac{8}{4} = 3,
\]
which achieves the theoretical upper bound.

As seen above, when the supports of minimum-weight codewords form a specific combinatorial structure, the upper bound on \( \lambda_1^{\max} \) becomes achievable. This observation is formalized in the following theorem.

\begin{theorem}\label{thm:systematic}
Consider a system encoded by a systematic generator matrix \( \boldsymbol{G} \) (over \( \mathbb{F}_2 \)). Then, for any object \( o_{\ell} \) corresponding to the systematic column \( \boldsymbol{c}^{\ell} = \boldsymbol{e}_{\ell} \) (with \( \ell \in [k] \)), the maximal achievable rate \( \lambda_{\ell}^{\max} \) satisfies
\[
1 + J_{\ell} \,\le\, \lambda_{\ell}^{\max} \,\le\, 1 + \frac{n - 1}{d^\perp - 1},
\]
where \( d^\perp \) is the minimum distance of \( \mathcal{C}^\perp \), which is the dual code of the code generated by $\boldsymbol{G}$, assuming \( d^\perp > 1 \).

The upper bound is achieved if the collection
\[
\mathcal{N}_{\ell} = \left\{ \mathrm{Supp}(\boldsymbol{h}) \setminus \{\ell\} \,\middle|\, \boldsymbol{h} \in \mathcal{C}^\perp,\; \textsf{w}(\boldsymbol{h}) = d^\perp,\; h_{\ell} = 1 \right\}
\]
forms a \( 1 \)-design on the point set \( [n] \setminus \{\ell\} \); that is, the supports (excluding position \( \ell \)) of all minimum-weight dual codewords with a 1 in position \( \ell \) form a \( 1 \)-design.
\end{theorem}

\begin{proof}
The lower and upper bounds follow from Theorem~\ref{thm:general_bound_demand}, noting that when $\boldsymbol{G}$ is systematic, the smallest recovery set for each object consists of a single systematic coordinate, i.e., \( a = 1 \). The upper bound can also be proved using Theorem IV.6 in~\cite{DualCode:GianiraAS22} by setting all other requests $\lambda_{j}, j \neq \ell$ to 0.

To prove achievability of the upper bound, assume the sets in \( \mathcal{N}_{\ell} \) form a \( 1 \)-\((n-1,\, d^\perp - 1,\, d_c)\) design. Let \( \gamma = |\mathcal{N}_{\ell}| \) be the number of blocks. Since each block contains \( d^\perp - 1 \) points, double-counting the total number of point-block incidences in this 1-design yields the identity
\begin{equation}\label{eq:design_equality}
(n - 1) \cdot d_c = \gamma \cdot (d^\perp - 1).
\end{equation}

For every \( \boldsymbol{h} \in \mathcal{C}^\perp \) such that $\boldsymbol{h}\setminus\{l\} \in  \mathcal{N}_{\ell}$, we have 
\( \boldsymbol{G} \cdot \boldsymbol{h}^\top = \boldsymbol{0}_k \). Since \( \boldsymbol{c}^{\ell} = \boldsymbol{e}_{\ell} \) and \( h_{\ell} = 1 \), it follows that:
\[
\sum_{j\, \in\, \mathrm{Supp}(\boldsymbol{h}) \setminus \{\ell\}} \boldsymbol{c}^j(\boldsymbol{G}) = \boldsymbol{0}_k - \boldsymbol{e}_{\ell} = \boldsymbol{e}_{\ell}.
\]
(This generalized Eq.~\eqref{eq:sumColumns}.) Thus, each set \( \mathrm{Supp}(\boldsymbol{h}) \setminus \{\ell\} \) is a recovery set of size $d^{\perp}-1$ for object \( o_{\ell} \). Because these recovery sets by assumption form a \( 1 \)-design, each server \( j \in [n] \setminus \{\ell\} \) appears in exactly \( d_c \) of the \( \gamma \) recovery sets. We now construct a fractional allocation
\[
\begin{cases}
\text{Assign } \lambda_{\ell} = 1 \text{ to the singleton set } \{\ell\}, \\
\text{Assign } \lambda_{\ell} = \dfrac{1}{d_c} \text{ to each of $\gamma$ recovery sets in } \mathcal{N}_{\ell}.
\end{cases}
\]
This allocation ensures that each server handles at most unit total demand ($d_c \cdot \frac{1}{d_c} = 1$). The total demand served for object \( o_{\ell} \) is
\[
\lambda_{\ell}^{\max} = 1 + \frac{\gamma}{d_c} = 1 + \frac{n - 1}{d^\perp - 1},
\]
where the final equality uses~\eqref{eq:design_equality}, thus matching the claimed upper bound.
\end{proof}

\medskip

We now observe that the achievability condition in Theorem~\ref{thm:systematic} is satisfied if the supports of the minimum-weight codewords of the dual code $\mathcal{C^{\perp}}$ form a \( 2 \)-design, yielding a more elegant and broadly applicable condition for achievability. This observation follows from the result below.

\begin{lemma}[Design Reduction, see, e.g.,{~\cite[Theorem 9.2]{CombiDesigns:books/daglib/Stinson04}}]
Let \( \mathcal{D} = (V, \mathcal{A}) \) be a \( t \)-design and let \( Z \subseteq V \) with \( |Z| = i < t \). Then the incidence structure \( (V \setminus Z,\, \{A \setminus Z : A \in \mathcal{A},\, Z \subset A\}) \) is a \( (t - i) \)-design.
\end{lemma}
That is, by fixing any subset \( Z \) of size \( i < t \), the substructures formed by removing \( Z \) from every block containing it yield a \( (t-i) \)-design on the remaining points.

Applying this Design Reduction result with \( t = 2 \) and \( Z = \{\ell\} \), we conclude that the subcollection of supports of minimum-weight dual codewords containing coordinate \( \ell \), punctured at position \( \ell \), forms a \( 1 \)-design on \( [n] \setminus \{\ell\} \), which confirms the achievability condition in Theorem~\ref{thm:systematic}. We thus have the following corollary.
\begin{corollary}\label{coro:2desgin}
Let \( \boldsymbol{G} \) be a systematic generator matrix of a binary linear code $\mathcal{C}$ with dual distance \( d^\perp > 1 \). If the supports of the minimum-weight codewords in the dual code \( \mathcal{C}^\perp \) form a \( 2 \)-design, then
\[
\lambda_{\ell}^{\max} \;=\; 1 + \frac{n - 1}{d^\perp - 1}, \quad \forall\, \ell \in [k].
\]
\end{corollary}
Conditions under which dual codewords of certain weights form block designs have been extensively studied, although much remains to be understood. Interested readers are referred to the Assmus--Mattson theorem~\cite{5_designs:journals/joct/AssusH69}, its related results~\cite{combinatorial_design:journals/tit/CalderbankDS91}, and references therein.

\begin{remark}
  For a code with a systematic generator matrix,
  Theorem~\ref{thm:systematic} relates each symbol’s
  \emph{maximal achevable rate} to its \emph{error-correction capability}.
  Combined with Proposition~\ref{prop:MJDecoding}, it follows that, for each coordinate $\ell\in[k]$, the largest number of errors a code can tolerate while still correctly recovering the message symbol $a_\ell$ using one-step MLD, denoted $t_{\ell}$, satisfies
  \[
     t_\ell \;=\;
     \left\lfloor \frac{J_\ell}{2} \right\rfloor
     \;\le\;
     \left\lfloor \frac{\lambda_{\ell}^{\max}-1}{2} \right\rfloor ,
  \]
  where $J_\ell$ denotes the number of parity-check sums orthogonal to
  coordinate~$\ell$. Moreover, Theorem~\ref{thm:systematic} directly implies the bound on $J_{\ell}$
  \[
     J_\ell \;\le\; \frac{n-1}{d^{\perp}-1},
  \]
  which is a classical bound that appears, for example, as Theorem~17 in
  \cite[Ch.~13]{Coding:books/MacWilliamsS77}.
\end{remark}

\subsection{Applications to Systematic Codes}
We now demonstrate that the upper bound in Theorem~\ref{thm:systematic} is matched for all systematic codes whose SRR is explicitly known--namely, Simplex codes, Maximum Distance Separable (MDS) codes, and repetition codes.
    \subsubsection{Simplex Codes} Consider a systematic \( [2^k - 1,\, k,\, 2^{k-1}] \) Simplex code. Its dual is a \( [2^k - 1,\, 2^k - k - 1,\, 3] \) Hamming code. The upper bound in Theorem~\ref{thm:systematic} becomes
    \[
    \lambda_{\ell}^{\max} \,\le\, 1 + \frac{2^k - 2}{3-1} = 2^{k-1}.
    \]
This bound is achieved by Corollary~\ref{coro:2desgin}, as the minimum-weight codewords in the dual Hamming code (i.e., weight-3 codewords) form a 2-design—specifically, a \emph{Steiner system}~\cite[Theorem 10.25]{CombiDesigns:books/daglib/Stinson04}, which confirms the tightness of the bound for Simplex codes. An alternative proof using graph-theoretic arguments appears in~\cite{SRR:conf/isit/KazemiKSS20, SRR:journals/tit/AktasJKKS21}.
    
\subsubsection{Systematic MDS Codes}  
  Consider a systematic $[n,k,n-k+1]$ MDS code with $n\ge k+1$.  
  Its dual is an $[n,n-k,k+1]$ MDS code, so the general upper bound reads
  \[
     \lambda_{\ell}^{\max}\;\le\;1+\frac{n-1}{k}.
  \]
The bound is achieved by Corollary~\ref{coro:2desgin}. Indeed, it is known that the minimum–weight codewords of an $[n,k,d_{\min}]$ MDS code form a $d_{\min}$-design~\cite{5_designs:journals/joct/AssusH69}. Since here $d_{\min}=n-k+1\ge2$, Lemma~\ref{thm:lower_balance} guarantees that the supports of the minimum-weight dual codewords always constitute a $2$-design. Hence the bound is tight for all systematic MDS codes. A purely graph–theoretic derivation appears in
  \cite{SRR:journals/tit/AktasJKKS21,SRR:lySV2025}.

    \subsubsection{Repetition and Single parity-check Codes} For the (systematic) repetition code with parameters \( [n,\, 1,\, n] \), the maximal achevable rate for the sole object \( o_1 \) is 
    \[
    \lambda_1^{\max} \,=\, n.
    \]
    Its dual is a systematic single parity-check (SPC) code with parameters \( [n,\, n - 1,\, 2] \). For any object \( o_{\ell} \), \( \ell \in [n-1] \), the maximal achievable rate is
    \[
    \lambda_{\ell}^{\max} \,=\, 2.
    \]
    These results follow directly from the fact that both codes are systematic MDS codes, and thus the general tightness and achievability results for MDS codes apply.

\subsection{Application to Hamming Codes}
\begin{theorem}\label{thm:Hamming_demand}
For any systematic Hamming code with parameters \( [2^k - 1,\, 2^k - k - 1,\, 3] \), the maximal achevable rate for any object \( o_{\ell} \), \( \ell \in [k] \), is
\[
\lambda_{\ell}^{\max} \,=\, 3.
\]
Therefore, increasing the number of server nodes does not improve the maximal achievable request rates for data objects in such systems.
\end{theorem}

\begin{proof}
Let \( \mathcal{C} \) denote the Hamming code. Its dual, \( \mathcal{C}^\perp \), is the \( [2^k - 1,\, k,\, 2^{k-1}] \) Simplex code, with minimum distance \( d^\perp = 2^{k-1} \). Substituting into the bound, we obtain
\begin{align*}
\lambda_{\ell}^{\max} &\,\le\, 1 + \frac{n - 1}{d^\perp - 1} = 1 + \frac{2^k - 2}{2^{k-1} - 1} = 3
\end{align*}

To show that the bound is tight for all \( \ell \in [k] \), note that
\begin{itemize}
    \item All nonzero codewords in the Simplex code \( \mathcal{C}^\perp \) have weight \( 2^{k-1} \), and hence are minimum-weight codewords.
    \item The supports of these codewords form a \( 2 \)-design~\cite{t_design:journals/ccds/DingT21}.
\end{itemize}

Thus, the achievability condition in~Corollary~\ref{coro:2desgin} is satisfied, and the upper bound is attained.
\end{proof}
\section{Conclusion}
We presented a unified framework for analyzing the service-rate region (SRR) of linear codes through the lens of orthogonal parity checks and combinatorial design theory. By establishing upper and lower bounds on the maximal achievable rate for each data object, we linked SRR analysis with the error-correction capability of one-step MLD. Our results recover known SRR characterizations for various codes and provide the first SRR description for Hamming codes. Furthermore, we showed that the achievability of the upper bound is guaranteed when minimum-weight dual codewords form a 2-design. This connection between SRR structure and classical design theory opens new directions for understanding and optimizing coded storage systems.
\section*{Acknowledgment}
This work was supported in part by NSF-BSF
grant FET-2120262. The authors thank Andrea Di Giusto for helpful discussions. 

\balance
\bibliography{bibliography_ISIT}
\bibliographystyle{IEEEtran}

\begin{IEEEbiography}{Michael Shell}
Biography text here.
\end{IEEEbiography}

\end{document}